\documentclass[12pt,draft,onecolumn,journal]{IEEEtran}
\usepackage{amssymb}
\usepackage[reqno]{amsmath}
\usepackage{amsthm}
\usepackage{amsfonts}
\usepackage{subfigure}
\usepackage{times}
\usepackage[final]{graphicx}
\usepackage{times,amsmath,epsfig}
\usepackage{upgreek}
\usepackage{latexsym,amssymb,mathrsfs}

\usepackage{algorithmicx}
\usepackage[ruled]{algorithm}
\usepackage{algpseudocode}
\usepackage{algpascal}
\usepackage{algc}
\usepackage{xcolor}
\usepackage{cite}
\newtheorem{thm}{Theorem}
\newtheorem{defi}{Definition}
\newtheorem{prop}{Proposition}
\newtheorem{lem}{Lemma}

\newtheorem{rmk}{Remark}

\IEEEoverridecommandlockouts
\begin{document}
\title{On the Simulatability Condition in Key Generation Over a Non-authenticated Public Channel}
\author{Wenwen Tu and Lifeng Lai\thanks{W. Tu and L. Lai are with Department of Electrical and Computer Engineering, Worcester Polytechnic Institute, Worcester, MA. Email: \{wtu, llai\}@wpi.edu. The work of W. Tu and L. Lai was supported by the National Science Foundation CAREER Award under Grant CCF-1318980 and by the National Science Foundation under Grant CNS-1321223. }}
\date{\today}
\maketitle
\pagestyle{plain}
\begin{abstract}
Simulatability condition is a fundamental concept in studying key generation over a non-authenticated public channel, in which Eve is active and can intercept, modify and falsify messages exchanged over the non-authenticated public channel. Using this condition, Maurer and Wolf showed a remarkable ``all or nothing'' result: if the simulatability condition does not hold, the key capacity over the non-authenticated public channel will be the same as that of the case with a passive Eve, while the key capacity over the non-authenticated channel will be zero if the simulatability condition holds. However, two questions remain open so far: 1) For a given joint probability mass function (PMF), are there efficient algorithms (polynomial complexity algorithms) for checking whether the simulatability condition holds or not?; and 2) If the simulatability condition holds, are there efficient algorithms for finding the corresponding attack strategy? In this paper, we answer these two open questions affirmatively. In particular, for a given joint PMF, we construct a linear programming (LP) problem and show that the simulatability condition holds \textit{if and only if} the optimal value obtained from the constructed LP is zero. Furthermore, we construct another LP and show that the minimizer of the newly constructed LP is a valid attack strategy. Both LPs can be solved with a polynomial complexity.
%In this paper, the problem of efficiently checking the simulatability condition for key generation in the presence of an active adversary with side information is considered. In order to generate a secret key, the legitimate terminals can discuss with each other over a noiseless but unauthenticated channel. The active adversary can interrupt and falsify the public discussion, using some strategy, to cheat the legitimate terminals. Checking simulatability condition is very important since the key capacity can be same as it is when the adversary is passive if the simulatability condition does not hold and on the other hand, if the simulatability condition holds, the key capacity is zero. However, the simulatability condition is a "priori" since no efficient protocol has been proposed to check it. We study the relationship between simulatability condition and linear programming, and propose an efficient protocol. The proposed protocol can also be used for the case of multiple terminals.
\end{abstract}
%%%%%%%%%%%%%%% Article Body %%%%%%%%%%%%%%%%%%%%%%%%%%%%%%%%%%%%%%%%%
\section{Introduction}
The problem of secret key generation via public discussion under both source and channel models has attracted significant research interests\cite{Maurer:TIT:93,Ahlswede:TIT:93,Csiszar:TIT:00,Nitinawarat:TIT:10,Chan:TIT:10,Csiszar:TIT:04,Ye:TIT:12,Maurer:EUROC:97,Maurer:TIT:03,Maurer:TIT:031,Maurer:TIT:032}. Under the source model, users observe correlated sources generated from a certain joint probability mass function (PMF), and can discuss with each other via a noiseless public channel. Any discussion over the public channel will be overheard by Eve. Furthermore, the public channel can either be authenticated or non-authenticated. An authenticated public channel implies that Eve is a passive listener. On the other hand, a non-authenticated public channel implies that Eve is active and can intercept, modify or falsify any message exchanged through the public channel.
%The problem of generating a secret key in the presence of an active adversary attracts a significant interests\cite{Maurer:TIT:03,Maurer:TIT:031,Maurer:EUROC:1997,yakovlev:TIT:2008}. Under source model, when the adversary is passive, the legitimate terminals, who observe correlated observations, can share a common secret key based on the public messages via an authenticated public channel, while the public discussion leaks negligible amount of information about the generated key\cite{Maurer:TIT:1993,Ahlswede:TIT:93}. Csisz$\acute{a}$r\cite{Csiszar:TIT:04} and Ye\cite{Ye:TIT:2012} have extended the problem into the case of multiple terminals and presented several significant results.

Clearly, the secret key rate that can be generated using the non-authenticated public channel is no larger than that can be generated using the authenticated pulic channel. In \cite{Maurer:EUROC:97,Maurer:TIT:03,Maurer:TIT:031,Maurer:TIT:032}, Maurer and Wolf introduced a concept of simulatability condition (this condition will be defined precisely in the sequel) and established a remarkable ``all or nothing'' result. In particular, they showed that for the secret key generation via a non-authenticated public channel with two legitimate terminals in the presence of an active adversary: 1) if the simulatability condition holds, the two legitimate terminals will not be able to establish a secret key, and hence the key capacity is 0; and 2) if the simulatability condition does not hold, the two legitimate terminals can establish a secret key and furthermore the key capacity will be the same as that of the case when Eve is passive. Intuitively speaking, if the simulatability condition holds, from its own source observations, Eve can generate fake messages that are indistinguishable from messages generated from legitimate users. On the other hand, if the simulatability condition does not hold, the legitimate users will be able to detect modifications made by Eve.

% Unfortunately, the problem of how to efficiently check this condition has been left open for a long time since no efficient protocol has been proposed for the simplest two terminals case, let alone for the case of multiple terminals.

It is clear that the simulatability condition is a fundamental concept for the key generation via a non-authenticated public channel, and hence it is important to design efficient algorithms to check whether the simulatability condition holds or not. Using ideas from mechanical models, \cite{Maurer:TIT:031} made significant progress in designing efficient algorithms. In particular, \cite{Maurer:TIT:031} proposed to represent PMFs as mass constellations in a coordinate, and showed that the simulatability condition holds if and only if one mass constellation can be transformed into another mass constellation using a finite number of basic mass operations. Furthermore, \cite{Maurer:TIT:031} introduced another notion of one mass constellation being ``more centered'' than another constellation and designed a low-complexity algorithm to check this ``more centered'' condition. %When the variable observed by one of the legitimate users is binary,
For some important special cases, which will be described precisely in Section~\ref{sec:pre}, \cite{Maurer:TIT:031} showed that the ``more centered'' condition is necessary and sufficient for the mass constellation transformation problem (and hence is necessary and sufficient condition for the simulatability condition for these special cases). However, in the general case, the ``more centered'' condition is a necessary but not sufficient condition for the mass constellation transformation problem. Hence, whether there exists efficient algorithms for the mass constellation transformation problem (and hence the simulatability condition) in the general case is still an open question.  %However, the method of comparing \textit{more centered} problem lose its effectiveness if $Y$ is a general case. It's a necessary, but not a sufficient condition for simulatability. %This criterion is complete for some special cases, however, it is a necessary condition, but it's not a sufficient condition for the rest of the cases. Let's say given a probability mass distribution $P_{XYZ}$, we would like to check whether $X$ is simulatable by $Z$ with respect to $Y$, this criterion works perfectly for the case when $Y$ is a binary variable. But for the general case, when $Y$ is not a binary variable, this criterion merely points out that if the \textit{More Centered} condition, defined by definition 3 in \cite{Maurer:TIT:031}, does not hold, the simulatability condition is not satisfied, it does not work if the \textit{More Centered} condition holds.

As the result, despite the significant progress made in~\cite{Maurer:TIT:031}, the following two questions remain open regarding the simulatability condition for the general case: \begin{enumerate}
\item For a given joint PMF, are there efficient algorithms (polynomial complexity algorithms) for checking whether the simulatability condition holds or not?
\item If the simulatability condition holds, are there efficient algorithms for finding the corresponding Eve's attack strategy?
\end{enumerate}
In this paper, we answer these two open questions affirmatively.

To answer the first open question, we construct a linear programming (LP) problem and show that the simulatability condition holds \textbf{if and only if} the optimal value obtained from this LP is zero. We establish our result in three main steps. We first show that, after some basic transformations, checking whether the simulatability condition holds or not is equivalent to checking whether there exists a nonnegative solution to a specially constructed system of linear equations. We then use a basic result from linear algebra to show that whether there exists a nonnegative solution to the constructed system of linear equations is equivalent to whether there is a solution (not necessarily nonnegative) to a related system of inequalities or not. Finally, we use Farkas' lemma~\cite{Alex:san:1998}, a fundamental result in linear programming and other optimization problems, to show that whether the system of inequalities has a solution or not is equivalent to whether the optimal value of a specially constructed LP is zero or not. Since there exists polynomial complexity algorithms for solving LP problems\cite{karmarkar:new:84,Gonzaga:PMP:05,Bazaraa:JWS:2011}, we thus find a polynomial complexity algorithm for checking the simulatability condition for a general PMF.
%In this paper, we propose a succinct protocol which works efficiently for the simulatability condition of general case, to completely solve the open problem. we first transform the problem into linear system of equations, and then perform some equivalent transformations to the linear system so that it becomes a linear programming problem. After transformation it can be checked whether the linear programming has solutions or not and get an optimal solution with a designed objective function.

To answer the second open question, we construct another LP and show that the minimizer of this LP is a valid attack strategy. The proposed approach is very flexible in the sense that one can simply modify the cost function of the constructed LP to obtain different attack strategies. Furthermore, the cost function can be modified to satisfy various design criteria. For example, a simple cost function can be constructed to minimize the amount of modifications Eve needs to perform during the attack. All these optimization problems with different cost functions can be solved with a polynomial complexity.

The remainder of the paper is organized as follows. In Section \ref{sec:pre}, we introduce some preliminaries and the problem setup. In Section \ref{sec:MR}, we present our main results. In Section \ref{sec:eple}, we use numerical examples to illustrate the proposed algorithm. In Section \ref{sec:recom}, we present an approach to further reduce the computational complexity. In Section \ref{sec:conc}, we offer our concluding remarks.

\section{Preliminaries and Problem Setup }\label{sec:pre}

Let $\mathcal{X}=\{1,\cdots,|\mathcal{X}|\}$, $\mathcal{Y}=\{1,\cdots,|\mathcal{Y}|\}$ and $\mathcal{Z}=\{1,\cdots,|\mathcal{Z}|\}$ be three finite sets. Consider three correlated random variables $(X,Y,Z)$, taking values from $\mathcal{X}\times\mathcal{Y}\times\mathcal{Z}$, with joint PMF $P_{XYZ}$, the simulatability condition is defined as follows:
\begin{defi}{(\cite{Maurer:EUROC:97})}
For a given $P_{XYZ}$, we say $X$ is simulatable by $Z$ with respect to $Y$, denoted by $\text{Sim}_Y(Z\rightarrow X)$, if there exists a conditional PMF $P_{\bar X|Z}$ such that $P_{Y \bar X }=P_{YX}$, with
\begin{eqnarray}
P_{Y \bar X}(y,x)=\sum_{z\in \mathcal Z} P_{YZ}(y,z)\cdot P_{\bar X|Z}(x|z), \label{eq:pxy}
\end{eqnarray}
in which $P_{YX}$ and $P_{YZ}$ are the joint PMFs of $(Y,X)$ and $(Y,Z)$ under $P_{XYZ}$ respectively.
\end{defi}

One can also define $\text{Sim}_X(Z\rightarrow Y)$ in the same manner. This concept of simulatability, first defined in~\cite{Maurer:EUROC:97}, is a fundamental concept in the problem of secret key generation over a non-authenticated public channel~\cite{Maurer:TIT:03,Maurer:TIT:031,Maurer:TIT:032}, in which two terminals Alice and Bob would like to establish a secret key in the presence of an adversary Eve. These three terminals observe sequences $X^N$, $Y^N$ and $Z^N$ generated according to
\begin{eqnarray}
P_{X^NY^NZ^N}(x^N,y^N,z^N)=\prod_{i=1}^NP_{XYZ}(x_i,y_i,z_i).
\end{eqnarray}
Alice and Bob can discuss with each other via a public \textbf{non-authenticated} noiseless channel, which means that Eve not only has full access to the channel but can also interrupt, modify and falsify messages exchanged over this public channel. The largest key rate that Alice and Bob\footnote{Please see ~\cite{Maurer:TIT:03,Maurer:TIT:031,Maurer:TIT:032} for precise definitions.} can generate with the presence of the active attacker is denoted as $S^*(X;Y||Z)$. Let $S(X;Y||Z)$ denote the largest key rate that Alice and Bob can generate when Eve is passive, i.e., when the public channel is authenticated. Clearly, $S(X;Y||Z)\geq S^*(X;Y||Z)$. Although a full characterization of $S(X;Y||Z)$ is unknown in general, \cite{Maurer:TIT:03} established the following remarkable ``all or nothing'' result:

\begin{thm}{(\cite{Maurer:TIT:03})}
If $\text{Sim}_{Y}(Z\rightarrow X)$ or $\text{Sim}_{X}(Z\rightarrow Y)$, then $S^*(X;Y||Z)=0$. Otherwise, $S^*(X;Y||Z)=S(X;Y||Z)$.
\end{thm}

This significant result implies that, if the simulatability condition does not hold, one can generate a key with the same rate as if Eve were passive. On the other hand, if the simulatability condition holds, the key rate will be zero. Intuitively speaking, if $\text{Sim}_Y(Z\rightarrow X)$ holds, then after observing $Z^N$, Eve can generate $\bar{X}^N$ by passing $Z^N$ through a channel defined by $P_{\bar{X}|Z}$. Then $(\bar{X}^N,Y^N)$ has the same statistics as $(X^N,Y^N)$. Hence by knowing only $Y^N$, Bob cannot distinguish $\bar{X}^N$ and $X^N$, and hence cannot distinguish Alice or Eve.

As mentioned in the introduction, \cite{Maurer:TIT:031} has made important progress in developing low-complexity algorithms for checking whether $\text{Sim}_{Y}(Z\rightarrow X)$ (or $\text{Sim}_{X}(Z\rightarrow Y)$) holds or not. In particular, \cite{Maurer:TIT:031} developed an efficient algorithm to check a related condition called ``more centered'' condition. When $|\mathcal{Y}|=2$, that is when $Y$ is a binary random variable, this ``more centered'' condition is shown to be necessary and sufficient for $\text{Sim}_{Y}(Z\rightarrow X)$. Hence, \cite{Maurer:TIT:031} has found an efficient algorithm to check $\text{Sim}_{Y}(Z\rightarrow X)$ for the special case of $Y$ being binary (the algorithm is also effective in checking $\text{Sim}_{X}(Z\rightarrow Y)$ when $X$ is binary). However, when $Y$ is not binary, the ``more centered'' condition is only a necessary condition for $\text{Sim}_{Y}(Z\rightarrow X)$. Hence, two questions remain open:
\begin{enumerate}
\item For a general given $P_{XYZ}$, are there efficient algorithms (polynomial complexity algorithms) for checking whether $\text{Sim}_{Y}(Z\rightarrow X)$ (or $\text{Sim}_{X}(Z\rightarrow Y)$) holds or not?
\item If $\text{Sim}_{Y}(Z\rightarrow X)$ (or $\text{Sim}_{X}(Z\rightarrow Y)$) holds, are there efficient algorithms for finding the corresponding $P_{\bar{X}|Z}$ (or $P_{\bar{Y}|Z}$)?
\end{enumerate}

In this paper, we answer these two open questions affirmatively.

\textit{Notations:} Throughout this paper, we use boldface uppercase letters to denote matrices, boldface lowercase letters to denote vectors. We also use $\mathbf 1$, $\mathbf 0$ and $\mathbf I$, unless stated otherwise, to denote all ones column vector, all zeros column vector and the identity matrix, respectively. In addition, we denote the vectorization of a matrix by Vec$(\cdot)$. Specifically, for an $m\times n$ matrix $\mathbf A$, Vec$(\mathbf A)$ is an $mn\times 1$ column vector:
\begin{eqnarray}
\text{Vec}(\mathbf A)=[a_{11},\cdots,a_{m1},\cdots,a_{1n},\cdots,a_{mn}]^T,
\end{eqnarray}
in which $[\cdot]^T$ is the transpose of the matrix. And vice versa can be done by $\mathbf A=$Reshape$(\text{Vec}(\mathbf A),[m,n])$.
We use $\mathbf A\otimes \mathbf B$ to denote the Kronecker product of matrices $\mathbf A$ and $\mathbf B$. Specifically, assume $\mathbf A$ is an $m\times n$ matrix, then
\begin{eqnarray}
\mathbf A\otimes \mathbf B=\left[
                \begin{array}{ccc}
                  a_{11}\mathbf B & \cdots & a_{1n}\mathbf B \\
                  \vdots & \ddots & \vdots \\
                  a_{m1}\mathbf B & \cdots & a_{mn}\mathbf B \\
                \end{array}
              \right].
\end{eqnarray}
All matrices and vectors in this paper are real.

\section{Main Results}\label{sec:MR}

In this paper, we focus on $\text{Sim}_Y(Z\rightarrow X)$. The developed algorithm can be easily modified to check $\text{Sim}_X(Z\rightarrow Y)$. We rewrite~\eqref{eq:pxy} in the following matrix form
\begin{eqnarray}
\mathbf C=\mathbf A \mathbf Q, \label{eq:xy1}
\end{eqnarray}
in which $\mathbf C=[c_{ij}]$ is a $|\mathcal Y|\times |\mathcal X|$ matrix with $c_{ij}=P_{YX}(i,j)$, $\mathbf A=[a_{ik}]$ is a $|\mathcal Y|\times |\mathcal Z|$ matrix with $a_{ik}=P_{YZ}(i,k)$, and $\mathbf Q=[q_{kj}]$ is a $|\mathcal Z|\times |\mathcal X|$ matrix with $q_{kj}=P_{\bar X|Z}(j|k)$ if such $P_{\bar X|Z}$ exists.

Checking whether $\text{Sim}_Y(Z\rightarrow X)$ holds or not is equivalent to checking whether there exists a transition matrix $\mathbf Q$ such that~\eqref{eq:xy1} holds. As $\mathbf Q$ is a transition matrix, its entries $q_{kj}$s must satisfy
\begin{eqnarray}
q_{kj}\geq 0,&&\quad \forall k\in [1:|\mathcal Z|], j\in[1:|\mathcal X|],\label{eq:xy2}\\
\sum_{j=1}^{|\mathcal X|}q_{kj}=1,&&\quad \forall k\in [1:|\mathcal Z|].\label{eq:xy3}
\end{eqnarray}
We note that if $q_{kj}$s satisfy~\eqref{eq:xy2} and~\eqref{eq:xy3}, they will automatically satisfy $q_{kj}\leq 1$. Hence, we don't need to state this requirement here.

If there exists at least one transition matrix $\mathbf Q$ satisfying \eqref{eq:xy1}, \eqref{eq:xy2} and \eqref{eq:xy3} simultaneously, we can conclude that the simulatability condition $\text{Sim}_Y(Z\rightarrow X)$ holds.

\eqref{eq:xy3} can be written in the matrix form
%\begin{eqnarray}
%\textbf 1_{|\mathcal Z|\times 1}=\left[
%              \begin{array}{ccc}
%                q_{11} & \cdots & q_{1|\mathcal X|} \\
%                \vdots & \ddots & \vdots\\
%                q_{|\mathcal Z|1} & \cdots & q_{|\mathcal Z||\mathcal X|} \\
%              \end{array}
%            \right]\cdot\textbf 1_{|\mathcal X|\times 1},\label{eq:xy6}
%\end{eqnarray}
\begin{eqnarray}
\mathbf 1_{|\mathcal Z|\times 1}=\mathbf{Q}\mathbf 1_{|\mathcal X|\times 1},\label{eq:xy6}
\end{eqnarray}
Then, \eqref{eq:xy1} and \eqref{eq:xy6} can be written in the following compact form:
\begin{eqnarray}
&&\hspace{-10mm}\left[
  \begin{array}{c}
    \text{Vec}({\mathbf C}^T)\\
    \mathbf {1}_{|\mathcal Z|\times 1} \\
  \end{array}
\right]\nonumber \\&=&\left[
  \begin{array}{cccc}
    a_{11} \mathbf I & a_{12} \mathbf I & \cdots & a_{1|\mathcal Z|} \mathbf I \\
    \vdots & \vdots & \ddots & \vdots \\
    a_{|\mathcal Y|1} \mathbf I & a_{|\mathcal Y|2} \mathbf I & \cdots & a_{|\mathcal Y||\mathcal Z|} \mathbf I \\
    \mathbf 1 & \mathbf 0 & \cdots & \mathbf 0 \\
    \mathbf 0 & \mathbf 1 & \ddots & \vdots \\
    \vdots & \ddots & \ddots & \vdots \\
    \mathbf 0 & \cdots & \mathbf 0 & \mathbf 1 \\
  \end{array}
\right]\text{Vec}({\mathbf Q}^T) \nonumber\\
&=&\left[
     \begin{array}{c}
       \mathbf A\otimes \mathbf I \\
        \mathbf I_{|\mathcal Z|}\otimes \mathbf 1 \\
     \end{array}
   \right]\text{Vec}({\mathbf Q}^T),
\label{eq:xy7}
\end{eqnarray}
in which the sizes for $\mathbf I$, $\mathbf 1$ and $\mathbf 0$ are $|\mathcal X|\times |\mathcal X|$, $1\times |\mathcal X|$ and $1\times |\mathcal X|$, respectively.
%%%%%%$\textbf X_i, i\in [1:|\mathcal Z|]$ is a row vector from matrix $\textbf{X}_{|\mathcal Z|\times |\mathcal X|}$, and $\textbf {1}$ is a row vector with all coordinates $1$.

For notational convenience, we define
   \begin{eqnarray}
&&   {\mathbf c}\triangleq\left[
  \begin{array}{c}
    \text{Vec}({\mathbf C}^T)\\
    \mathbf {1}_{|\mathcal Z|\times 1} \\
  \end{array}
\right],\label{eq:c}\\
&&\mathbb A\triangleq\left[
     \begin{array}{c}
       \mathbf A\otimes \mathbf I \\
        \mathbf I_{|\mathcal Z|}\otimes \mathbf 1 \\
     \end{array}
   \right],\label{eq:A}\\
&&   {\mathbf q}\triangleq\text{Vec}({\mathbf Q}^T). \label{eq:q}
   \end{eqnarray}
From~\eqref{eq:xy7}, it is clear that ${\mathbf c}$ is an $m\times 1$ vector, $\mathbb A$ is an $m\times n$ matrix, and ${\mathbf q}$ is an $n\times 1$ vector, in which
\begin{eqnarray}
m&=&|\mathcal Y||\mathcal X|+|\mathcal Z|,\label{eq:m}\\
n&=&|\mathcal Z||\mathcal X|\label{eq:n}.
\end{eqnarray}

With these notation and combining \eqref{eq:xy7} with \eqref{eq:xy2}, the original problem of checking whether $\text{Sim}_{Y}(Z\rightarrow X)$ holds or not is equivalent to checking whether there exists \textbf{nonnegative} solutions $\mathbf q$ for the system
\begin{eqnarray}
\mathbb A \mathbf q=\mathbf c. \label{eq:xy8}
\end{eqnarray}

In the following, we check whether there exists at least a nonnegative solution for the system defined by \eqref{eq:xy8}. There are two main steps: 1) whether the system is consistent or not; 2) if it is consistent, whether there exists a nonnegative solution or not. Checking the consistency of~\eqref{eq:xy8} is straightforward: a necessary and sufficient condition for a system of non-homogenous linear equations to be consistent is
\begin{eqnarray}
\text{Rank}(\mathbb A)=\text{Rank}((\mathbb A|\mathbf c)), \label{eq:xy9}
\end{eqnarray}
where $(\mathbb A|\mathbf c)$ is the augmented matrix of $\mathbb A$.
If \eqref{eq:xy9} is not satisfied, it can be concluded that $\text{Sim}_Y(Z\rightarrow X)$ does not hold. If \eqref{eq:xy9} is satisfied, we need to further check whether there exists a nonnegative solution to~\eqref{eq:xy8} or not.

To proceed further, we will need the following definition of generalized inverse (g-inverse) of a matrix $\mathbf G$.

\begin{defi}{(\cite{Rao:NYW:1971})}
For a given $m\times n$ real matrix $\mathbf G$, an $n\times m$ real matrix $\mathbf G^g$ is called a g-inverse of $\mathbf G$ if $$\mathbf G\mathbf G^g\mathbf G=\mathbf G.$$
\end{defi}

The g-inverse ${\mathbf G}^{g}$ is generally not unique (If $n=m$ and ${\mathbf G}$ is full rank, then ${\mathbf G}^{g}$ is unique and equal to the inverse matrix ${\mathbf G}^{-1}$). A particular choice of g-inverse is called the Moore-Penrose pseudoinverse $\mathbf{G}^+$, which can be computed using multiple different approaches. One approach is to use the singular value decomposition (SVD): by SVD, for a given $\mathbf G$ and its SVD decomposition
\begin{eqnarray}
{\mathbf G}=\mathbf U {\mathbf\Sigma}\mathbf V^T,\label{eq:xy10}
\end{eqnarray}
then, ${\mathbf G}^+$ can be obtained as
\begin{eqnarray}
{\mathbf G}^+=\mathbf V{\mathbf\Sigma}^+\mathbf U^{T},\label{eq:xy13}
\end{eqnarray}
in which ${\mathbf\Sigma}^+$ is obtained by taking the reciprocal of each non-zero element on the diagonal of the diagonal matrix ${\mathbf\Sigma}$, leaving the zeros in place. One can easily check that the Moore-Penrose pseudoinverse $\mathbf{G}^+$ obtained by SVD satisfies the g-inverse matrix definition and hence is a valid g-inverse.

With the concept of g-inverse, we are ready to state our main result regarding the first open question.

\begin{thm}  \label{thm:2}
Let $\mathbb A^g$ be any given g-inverse of $\mathbb A$ (e.g., it can be chosen as the Moore-Penrose pseudoinverse $\mathbb A^+$), and $h^*$ be obtained by the following LP
\begin{eqnarray}
&& h^*=\min\limits_{{\mathbf{t}}}\{\mathbf{t}^T\mathbb A^g \textbf c\}, \label{eq:xy15}\\
\text{s. t. }&& \mathbf{t}\succeq \mathbf 0,\nonumber\\
&& (\mathbf I-\mathbb A^g\mathbb A)^T\mathbf{t}=\mathbf 0.\nonumber
\end{eqnarray}
Then $\text{Sim}_{Y}(Z\rightarrow X)$ holds, \textbf{if and only if} $h^*=0$ and~\eqref{eq:xy9} holds.
\end{thm}
\begin{proof}
If~\eqref{eq:xy9} does not hold, then there is no solution to~\eqref{eq:xy8}, and hence $\text{Sim}_{Y}(Z\rightarrow X)$ does not hold.

In the remainder of the proof, we assume that~\eqref{eq:xy9} holds. If~\eqref{eq:xy9} holds, the general solution to~\eqref{eq:xy8} can be written in the following form (see, e.g., Theorem 2 a.(d) of~\cite{Rao:san:1967})
\begin{eqnarray}
\mathbf q=\mathbb A^g \mathbf c+(\mathbb A^g\mathbb A-\mathbf I) \mathbf p, \label{eq:xy11}
\end{eqnarray}
in which $\mathbb A^g$ can be any given g-inverse of $\mathbb A$, and $\mathbf p$ is an arbitrary length-$n$ vector.

As the result, the problem of whether there exists a nonnegative solution to~\eqref{eq:xy8} (i.e., ${\mathbf q}\succeq {\mathbf 0}$) is equivalent to the problem of whether there exists a solution $\mathbf p$ for the following system defined by
\begin{eqnarray}
(\mathbf I-\mathbb A^g\mathbb A)\mathbf p\preceq \mathbb A^g \mathbf c  \label{eq:xy12}.
\end{eqnarray}
To check whether the system defined by \eqref{eq:xy12} has a solution, we use Farkas' lemma, a fundamental lemma in linear programming and related area in optimization. For completeness, we state the form of Farkas' lemma used in our proof in Appendix~\ref{app:farkas}. To use Farkas' lemma, we first write a LP related to the system defined in \eqref{eq:xy12}
\begin{eqnarray}
&& h^*=\min\limits_{\mathbf t}\{\mathbf t^T\mathbb A^g\mathbf c\}, \nonumber\\
\text{s.t. }&& \mathbf t\succeq \mathbf 0,\nonumber \\
&&(\mathbf I-\mathbb A^g\mathbb A)^T\mathbf t=\textbf 0. \nonumber%\label{eq:xy15}
\end{eqnarray}

The above LP is always feasible since $\mathbf t =\mathbf 0$ is a vector that satisfies the constraints, which results in $\mathbf t^T\mathbb A^g\mathbf c=0$. Hence the optimal value $h^*\leq 0$. Using Farkas' lemma, we have that~\eqref{eq:xy12} has a solution \textbf{if and if} $h^*=0$. More specifically, if $h^*=0$, then there exists at least a solution $\mathbf p$ for~\eqref{eq:xy12}, which further implies that there is a nonnegative solution to~\eqref{eq:xy8}, and hence $\text{Sim}_{Y}(Z\rightarrow X)$ holds.
On the other hand, if $h^*<0$, then there is no solution $\mathbf p$ for~\eqref{eq:xy12}, which further implies that there is no nonnegative solution to~\eqref{eq:xy8}, and hence $\text{Sim}_{Y}(Z\rightarrow X)$ does not hold.
\end{proof}

As mentioned above, if $\text{Rank}(\mathbb A)=m=n$ holds, then $\mathbb A^g=\mathbb A^{-1}$ is unique. For other cases, $\mathbb A^g$ might not be unique. One may wonder whether different choices of $\mathbb A^g$ will affect the result in Theorem~\ref{thm:2} or not. The following proposition answers this question.
\begin{prop}\label{prop:1}
Different choices of $\mathbb A^g$ will not affect the result on whether $h^*$ equals 0 or not.
\end{prop}
\begin{proof}
 %, and $\textbf q=\mathbb A^g \textbf c$(In fact, we can check directly whether $\textbf q\succeq 0$ holds here.). LP\eqref{eq:xy15} is equivalent to
%\begin{eqnarray}
%&& h^*=\min\limits_{\textbf y}\{\textbf y\textbf q\}, \nonumber\\
%\text{s.t. }&& \textbf y^T\succeq \textbf 0\nonumber.
%\end{eqnarray}
%We still can judge whether $\textbf q\succeq 0$ holds by the value of $h^*$.
%If $\mathbb A$ is not a full rank matrix, $\mathbb A^g$ has more than one possible value.
Let $\mathbb A_1^g$ and $\mathbb A_2^g$ be two different g-inverses of $\mathbb A$, and let $h_1^*$ and $h_2^*$ be the values obtained using $\mathbb A_1^g$ and $\mathbb A_2^g$ in~\eqref{eq:xy15} respectively. It suffices to show that if $h_1^*=0$, then $h_2^*=0$.

Assuming that $h_1^*=0$, then there exists a vector $\mathbf p_1$ satisfying $(\mathbf I-\mathbb A_1^g\mathbb A)\mathbf p_1\preceq \mathbb A_1^g \textbf c$, we will show that there exists a vector $\mathbf p_2$ satisfying $ (\mathbf I-\mathbb A_2^g\mathbb A)\mathbf p_2\preceq \mathbb A_2^g \mathbf c$, which then implies $h_2^*=0$.

First, we know that $\mathbb A_1^g\mathbf c$ and $\mathbb A_2^g\mathbf c$ are two solutions to the system $\mathbb A\mathbf q=\mathbf c$, which can be easily verified by setting $\mathbb A^g$ as $\mathbb A_1^g$ and $\mathbb A_2^g$ in~\eqref{eq:xy11} respectively and setting $\mathbf p=\mathbf 0$. This implies that \begin{eqnarray}
\mathbb A(\mathbb A_2^g\mathbf c-\mathbb A_1^g\mathbf c)=\mathbf 0,
\end{eqnarray}
and hence $\mathbb A_2^g\mathbf c-\mathbb A_1^g\mathbf c$ is a solution to the system $\mathbb A\mathbf q=\mathbf 0$.

Second, we know that any solution to the system $\mathbb A\mathbf q=\mathbf 0$ can be written in the form $(\mathbf I-\mathbb A^g\mathbb A) \mathbf p$~\cite{Rao:san:1967}. As $\mathbb A_2^g\mathbf c-\mathbb A_1^g\mathbf c$ is a solution to system $\mathbb A\mathbf q=\mathbf 0$, there must exist a $\mathbf p_0$ such that
 \begin{eqnarray}
(\mathbf I-\mathbb A_2^g\mathbb A)\mathbf p_0= \mathbb A_2^g\mathbf c-\mathbb A_1^g\mathbf c.\label{eq:24}
  \end{eqnarray}
In addition, it is easy to check that $(\mathbf I-\mathbb A_1^g\mathbb A)\mathbf p_1+(\mathbf I-\mathbb A_2^g\mathbb A)\mathbf p_0$ is also a solution to the system $\mathbb A\mathbf q=\mathbf 0$. Thus, there exists a $\mathbf p_2$ such that
 \begin{eqnarray}
(\mathbf I-\mathbb A_2^g\mathbb A)\mathbf p_2= (\mathbf I-\mathbb A_1^g\mathbb A)\mathbf p_1+(\mathbf I-\mathbb A_2^g\mathbb A)\mathbf p_0.\label{eq:25}
  \end{eqnarray}

Plugging~\eqref{eq:24} into~\eqref{eq:25}, we have
\begin{eqnarray}
(\mathbf I-\mathbb A_2^g\mathbb A)\mathbf p_2&=& (\mathbf I-\mathbb A_1^g\mathbb A)\mathbf p_1+(\mathbf I-\mathbb A_2^g\mathbb A)\mathbf p_0\nonumber\\
&=&(\mathbf I-\mathbb A_1^g\mathbb A)\mathbf p_1+\mathbb A_2^g\mathbf c-\mathbb A_1^g\mathbf c\\
&\preceq & \mathbb A_2^g\mathbf c,
\end{eqnarray}
in which the last inequality comes from the assumption that $(\mathbf I-\mathbb A_1^g\mathbb A)\mathbf p_1\preceq \mathbb A_1^g \mathbf c$. Hence, we have found a $\mathbf p_2$, such that $(\mathbf I-\mathbb A_2^g\mathbb A)\mathbf p_2\preceq \mathbb A_2^g \mathbf c$. This implies that $h_2^*=0$.

%The converse can be shown easily. Hence, we can conclude that the different selections of $\mathbb A^g$ has no influence on the final result.
\end{proof}
\begin{rmk} \label{lem:1}
The proposed algorithm for checking whether $\text{Sim}_{Y}(Z\rightarrow X)$ holds or not has a polynomial complexity. Among all operations required, computing the g-inverse and solving the LP defined by \eqref{eq:xy15} require most computations. The complexity to obtain $\mathbb A^g$ is of order $O(n^3)$~\cite{Moller:exact:1999}.
%If we choose $\mathbb A^g$ to be the Moore-Penrose pseudoinverse $\mathbb A^+$ and use SVD to obtain $\mathbb A^+$, its complexity is of order $O(n^3)$ \cite{Courrieu:arXiv:2008}.
Furthermore, there exists polynomial complexity algorithms to solve the LP defined by \eqref{eq:xy15}. For example, \cite{Gonzaga:PMP:05} provided an algorithm to solve LP using $O(n^3L)$ operations, where $L$ is number of binary bits needed to store input data of the problem (one can refer to Chapter 8 in~\cite{Bazaraa:JWS:2011} for more details about the complexity of algorithms for solving LP). Hence, the total operations of our algorithm for checking $\text{Sim}_{Y}(Z\rightarrow X)$ is of order $O(n^3L)$. In addition, we note that we can terminate the LP algorithm earlier once the algorithm finds a $\mathbf t$ such that ${\mathbf t}{\mathbb A}^g{\mathbf c}<0$, as this indicates that $h^*<0$. This can potentially further reduce the computational complexity.
\end{rmk}

Thus, we can conclude that the proposed algorithm can check whether $\text{Sim}_{Y}(Z\rightarrow X)$ holds or not with a polynomial complexity. Algorithm~\ref{alg:1} summarizes the main steps involved in our algorithm. In the following algorithm, we use $\text{Res}=0$ to denote that $\text{Sim}_{Y}(Z\rightarrow X)$ does not hold and $\text{Res}=1$ to denote that $\text{Sim}_{Y}(Z\rightarrow X)$ holds.
%To sum up, the proposed protocol for checking the simulatability condition can be represented by the following flow chart in Fig.\ref{fg:1}.
%\begin{figure}[ht]
%\centering
%\includegraphics[width=0.3\textwidth]{flow.png}
%\quad
%\caption{Flow chart for the proposed protocol}\label{fg:1}
%\end{figure}

\alglanguage{pseudocode}
\begin{algorithm}[h]
\caption{Checking $\text{Sim}_{Y}(Z\rightarrow X)$}\label{alg:1}
\begin{algorithmic}[1]
%\Procedure {FindPathBK}{$v$, $u$, $p$}
\State \textbf{Input: }PMF $P_{XYZ}$;
\newline
\State \textbf{Initiate: }
\State \hspace{5mm}a. Calculate matrices $\mathbf{A}$ and $\mathbf{C}$;
\State \hspace{5mm}b. Construct $\mathbf c$ and $\mathbb A$ using~\eqref{eq:c} and~\eqref{eq:A} respectively;
\State \hspace{5mm}c. Set $\text{Res}=0$;
\newline
\If {($\text{Rank}(\mathbb A)\neq \text{Rank}(\mathbb A|\mathbf c)$)}
\State \textbf{break;}
\Else
\State d. Find a $\mathbb A^g$, and calculate $\mathbb A^g\mathbf c$, $\mathbf I-\mathbb A^g\mathbb A$;
\State e. Solve LP~\eqref{eq:xy15} and obtain $h^*$;
    \If {($h^*==0$)}
    \State $\text{Res}=1$;
    \Else
    \State \textbf{break;}
    \EndIf
\EndIf
\newline
\State \textbf{Output: } Res.
\end{algorithmic}
\end{algorithm}

In the following, we provide our answer to the second open question, i.e., if $\text{Sim}_Y(Z\rightarrow X)$ holds, how to find $P_{\bar{X}|Z}$ efficiently.
\begin{thm}
Let $\mathbf e$ be any $n\times 1$ vector with $\mathbf e\succ \textbf 0$, and $\mathbf q^*$ be the obtained from the following LP:
\begin{eqnarray}
&& \min\limits_{\mathbf q}f(\mathbf q)= \mathbf e^T\mathbf q, \label{eq:xy16}\\
\text{s.t. }&& \mathbf q\succeq \mathbf 0,\nonumber \\
&& \mathbb A\mathbf q=\mathbf c.\nonumber
\end{eqnarray}
If $\text{Sim}_Y(Z\rightarrow X)$ holds, then $\mathbf Q^*=\text{Reshape}(\mathbf q^*,[|\mathcal X|,|\mathcal Z|])^T$ is a valid choice for $P_{\bar X|Z}$.
\end{thm}
\begin{proof}
By assumption, $\text {Sim}_Y(Z\rightarrow X)$ holds, which implies that the system defined by \eqref{eq:xy8} is consistent and it has nonnegative solutions. Hence, the following LP is feasible
\begin{eqnarray}
&& \min\limits_{\mathbf q}f(\mathbf q)= \mathbf e^T\mathbf q,\\
\text{s.t. }&& \mathbf q\succeq \mathbf 0,\nonumber \\
&& \mathbb A\mathbf q=\mathbf c,\nonumber
\end{eqnarray}
where $\mathbf e\succ \mathbf 0$. Hence, the minimizer $\mathbf q^*$ is nonnegative and satisfies $\mathbb A\mathbf q^*=\mathbf c$. We can then reshape $\mathbf q^*$ into matrix $\mathbf Q^*$ (see~\eqref{eq:q}). $\mathbf Q^*$ is a valid choice for $P_{\bar{X}|Z}$.
\end{proof}
\begin{rmk}
Since finding a suitable $P_{\bar X|Z}$ using our approach is equivalent to solving a LP, the complexity is of polynomial order.
\end{rmk}

%The problem of find a solution $P_{\bar X|Z}$ is to do another LP. Thus, the required operations complexity is of order $O(n^3L)$.
\begin{rmk}
For a given distribution $P_{XYZ}$, there may be more than one possible $P_{\bar X|Z}$ such that~\eqref{eq:pxy} holds. Different choices of $\textbf e$ in~\eqref{eq:xy16} give different values for $P_{\bar X|Z}$. %We can see in the examples in next section, that with different $a_i,i\in [1:n]$ input, different values for $P_{\bar X|Z}$ will be returned, and all of them are feasible.
\end{rmk}

\begin{rmk}
The objective function $f(\mathbf q)$ can be further modified to satisfy various design criteria of Eve. For example, let $$\tilde{\mathbf q}=\text{Vec}(\tilde{\mathbf Q}[\tilde{q}_{kj}]^T)$$ with $\tilde{q}_{kj}=P_{X|Z}(k|j)$, then setting $$f(\mathbf q)=||\mathbf q-\tilde{\mathbf q}||_2^2$$ will minimize the amount of changes in the conditional PMF in the $l_2$ norm sense. This is a quadratic programming, which can still be solved efficiently.
\end{rmk}

\section{Numerical Examples}  \label{sec:eple}

In this section, we provide several examples to illustrate the proposed algorithm. We also use some of the examples used in \cite{Maurer:TIT:031} to compare our proposed algorithm with the method in \cite{Maurer:TIT:031}.

\textbf{Example 1:} Let $P_{XYZ}$ with ranges $\mathcal{X}=\{x_1, x_2\}$, $\mathcal{Y}=\{y_1, y_2\}$ and $\mathcal{Z}=\{z_1, z_2, z_3\}$ be:
\begin{eqnarray}
&&P_{XYZ}(x_1,y_1,z_1)=6/100,\nonumber \\
&&P_{XYZ}(x_2,y_1,z_1)=4/100,\nonumber \\
&&P_{XYZ}(x_1,y_1,z_2)=9/100,\nonumber \\
&&P_{XYZ}(x_2,y_1,z_2)=6/100,\nonumber \\
&&P_{XYZ}(x_1,y_1,z_3)=15/100,\nonumber \\
&&P_{XYZ}(x_2,y_1,z_3)=10/100,\nonumber \\
&&P_{XYZ}(x_1,y_2,z_1)=36/100,\nonumber \\
&&P_{XYZ}(x_2,y_2,z_1)=4/100,\nonumber \\
&&P_{XYZ}(x_1,y_2,z_2)=9/100,\nonumber \\
&&P_{XYZ}(x_2,y_2,z_2)=1/100,\nonumber \\
&&P_{XYZ}(x_1,y_2,z_3)=0,\nonumber \\
&&P_{XYZ}(x_2,y_2,z_3)=0.\nonumber
\end{eqnarray}

To use our algorithm, we have the following steps:\\
\textit{Step 1:} Compute $P_{YZ}$ and $P_{YX}$, and write them in the matrix form $\mathbf A$ and $\mathbf C$:
\begin{eqnarray}
\textbf A=\left[
    \begin{array}{ccc}
      0.1 & 0.15 & 0.25 \\
      0.4 & 0.1 & 0 \\
    \end{array}
  \right], \textbf C=\left[
               \begin{array}{cc}
                 0.3 & 0.2 \\
                 0.45 & 0.05 \\
               \end{array}
             \right].
\end{eqnarray}
\textit{Step 2:} Construct $\mathbb A$ and $\mathbf c$ using~\eqref{eq:c} and~\eqref{eq:A} respectively:
\begin{eqnarray}
&&\mathbb A=\left[
         \begin{array}{cccccc}
           0.1 & 0 & 0.15 & 0 & 0.25 & 0 \\
           0 & 0.1 & 0 & 0.15 & 0 & 0.25 \\
           0.4 & 0 & 0.1 & 0 & 0 & 0 \\
           0 & 0.4 & 0 & 0.1 & 0 & 0 \\
           1 & 1 & 0 & 0 & 0 & 0 \\
           0 & 0 & 1 & 1 & 0 & 0 \\
           0 & 0 & 0 & 0 & 1 & 1 \\
         \end{array}
       \right],\\
&&\mathbf c=[0.3, 0.2, 0.45, 0.05, 1, 1, 1]^T.
\end{eqnarray}
\textit{Step 3:} Check the ranks of $\mathbb A \text{ and }(\mathbb A|\mathbf c)$:

We get
\begin{eqnarray}
\text{Rank}(\mathbb A)=\text{Rank}((\mathbb A|\mathbf c))=5.
\end{eqnarray}
\textit{Step 4:} Choose the g-inverse to be the Moore-Penrose pseudoinverse $\mathbb A^+$ and calculate $\mathbb A^+\mathbf c$ and $\mathbf I-\mathbb A^+\mathbb A$:
\begin{eqnarray}
\mathbb A^+\mathbf c=\left[
         \begin{array}{c}
           0.9762 \\
           0.0238 \\
           0.5952 \\
           0.4048 \\
           0.4524 \\
           0.5476 \\
         \end{array}
       \right],
\end{eqnarray}
\begin{eqnarray}
  &&\hspace{-10mm}\mathbf I-\mathbb A^+\mathbb A=\nonumber \\&&\hspace{-10mm}\footnotesize{\left[
\begin{array}{cccccc}
                   \!\!0.0238\!\! &\!\! -0.0238\!\! &\!\! -0.0952\!\! &\!\! 0.0952 \!\!&\!\! 0.0476\!\! &\!\! -0.0476\!\! \\
                   \!\!-0.0238\! \!&\!\!0.0238\! \!&\!\! 0.0952\! \!&\! \!-0.0952\! \!&\! \!-0.0476 \!\!&\! \!0.0476\!\! \\
                   \!\!-0.0952 \!\!& \!\!0.0952 \!\!&\!\!0.3810 \! \!& \!\!-0.3810  \!\!&\!\! -0.1905 \!\!& \!\!0.1905\!\! \\
                   \!\!0.0952 \!\! &\! \! -0.0952 \!\! &\!\!  -0.3810 \!\!& \!\! 0.3810 \! \!&\! \!0.1905 \! \!&\! \! -0.1905\! \!\\
                  \! \!0.0476 \!\! & \! \!-0.0476\!\! & \! \! -0.1905\!\!& \!\!  -0.1905\! \!& \!\! 0.0952\! \!&  \!\! -0.0952 \!\!\\
                   \!\!-0.0476\! \! &\! \!0.0476\! \!& \! \!0.1905\! \! & \!\! -0.1905\! \!& \! \! -0.0952\!\! &\!  \!0.0952 \!\!\\
                 \end{array}
               \right].}
\end{eqnarray}
\textit{Step 5:} Solve LP~\eqref{eq:xy15}. Using the above data, we obtain $h^*=0$, which implies that $\text{Sim}_Y(Z\rightarrow X)$ holds.\\
\textit{Step 6:} Obtain a possible $P_{\bar X|Z}$.  We construct the LP defined in~\eqref{eq:xy16} with $\mathbf e=[2, 2 ,2 ,1, 1, 1]^T$, and get
$$\mathbf q^*=[1 ,0 ,1/2, 1/2 ,1/2 ,1/2]^T.$$ Thus the simulatability channel is
\begin{eqnarray}
P_{\bar X|Z}=\left[
              \begin{array}{cc}
                1 & 0 \\
                1/2 & 1/2 \\
                1/2 & 1/2 \\
              \end{array}
            \right],
\end{eqnarray}
which is consistent with the result obtained from the criterion proposed in \cite{Maurer:TIT:031}. If we set $\mathbf e=[1,1,1,1,1,1]^T$, we get $$\mathbf q^*=[0.9762 ,0.0238 ,0.5952, 0.4048 ,0.4524 ,0.5476]^T,$$ which implies that another valid choice is
\begin{eqnarray}
P_{\bar X|Z}=\left[
              \begin{array}{cc}
                0.9762 & 0.0238 \\
                0.5962 & 0.4048 \\
                0.4524 & 0.5476 \\
              \end{array}
            \right]. \label{eq:xz}
\end{eqnarray}

\textbf{Example 2:} In this example, we consider a case in which $Y$ is not binary. To represent the joint PMF concisely, we follow the same approach in~\cite{Maurer:TIT:031} and use $$M_{UV}=(P_U(u),(P_{V|U=u}(v_1),\cdots,P_{V|U=u}(v_{|\mathcal V|-1})))_{u\in \mathcal U}$$ to represent the joint PMF $P_{UV}$. For this example, we set
\begin{eqnarray}
M_{ZY}&=&(0.3,(0,0)),(0.3,(0.5,0)),\nonumber \\
&&(0.3,(0.25,\sqrt{3}/4)),(0.1,(0.25,\sqrt{3}/12)),\nonumber\\
M_{XY}&=&(0.3,(0.25,0)),(0.3,(0.375,\sqrt 3/8)),\nonumber \\&&(0.3,(0.125,\sqrt 3/8))(0.05,(0.24,\sqrt 3/12)) \nonumber \\&&(0.05,(0.26,\sqrt 3/12)).
\end{eqnarray}

In step 1, we write $P_{YZ}$ and $P_{YX}$ in the matrix form $\mathbf A$ and $\mathbf C$:
\begin{eqnarray}
&&\hspace{-10mm}\mathbf A=\left[
    \begin{array}{cccc}
            0  &  0.1500&    0.0750&    0.0250\\
         0      &   0    &0.1299   & 0.0144\\
    0.3000    &0.1500    &0.0951   & 0.0606\\
    \end{array}
  \right],\nonumber\\
&&\hspace{-10mm}\mathbf C=\left[
  \begin{array}{ccccc}
 0.0750  &  0.1125 &   0.0375 &   0.0120  &  0.0130\\
         0  &  0.0650 &   0.0650&    0.0072 &   0.0072\\
    0.2250  &  0.1225 &   0.1975 &   0.0308 &   0.0298\\
  \end{array}
\right].\nonumber
\end{eqnarray}

To make the paper concise, we do not list the values of $\mathbb A$, $\mathbf c$ and following steps in details. Steps $2,3,4$ are similar to those in Example $1$. But in Step $5$, we obtain that $h^*<0$, which indicates that $\text{Sim}_Y(Z\rightarrow X)$ does not hold. This result is also consistent with the conclusion in~\cite{Maurer:TIT:031}, which is obtained by an analysis that exploits the special mass constellation structure of the data. We note that the mechanical model based ``more centered'' criterion in \cite{Maurer:TIT:031} does not work for this example, as $Y$ is not binary anymore, although the mass constellation representation of PMFs can still be used to exploit the special structure that this set of data has.

Next, we provide an example for which the mass constellation presentation does not work while our algorithm can easily obtain the answers.

\textbf{Example 3:} In this example, we consider $X,Y,Z$ with larger dimensions, in particular, we set $|\mathcal{X}|=4$, $|\mathcal{Y}|=4$, and $|\mathcal{Z}|=6$. Again to represent the joint PMF concisely, we use the same method as that used in Example $2$ to represent $P_{XYZ}$. For this example, we randomly set

\footnotesize{
\begin{eqnarray}
&&\hspace{-10mm}\normalsize{M_{ZY}=}\nonumber\\
&&\hspace{-10mm}(0.1604,(0.1966,0.1054,0.4198)),(0.1654,(0.1230,0.4709,0.3355)),\nonumber\\
&&\hspace{-10mm}(0.1613,(0.0350,0.6219,0.0823)),(0.1504,(0.4585,0.2504,0.2343)),\nonumber\\
&&\hspace{-10mm}(0.1207,(0.2443,0.4704,0.0701)),(0.2419,(0.2979,0.1151,0.4601));\nonumber\\
&&\hspace{-10mm}\normalsize{M_{XY}=}\nonumber\\
&&\hspace{-10mm}(0.2603,(0.1784,0.3822,0.2056)),(0.2181,(0.1538,0.4409,0.2255)),\nonumber\\
&&\hspace{-10mm}(0.2356,(0.2129,0.2684,0.3913)),(0.2861,(0.3422,0.2044,0.3363)).\nonumber
\end{eqnarray}}
\normalsize{}

We denote the above PMF with following two matrices
\begin{eqnarray}
&&\hspace{-8mm}\mathbf{A}=\left[
             \begin{array}{cccccc}
                \! 0.0315 \! &  \!0.0203\! &\!   0.0056\! & \!  0.0690 \! &\!  0.0295\! & \!  0.0720\!\\
    \!0.0169\!  & \! 0.0779\! &\!   0.1003\! & \!  0.0377 \! & \! 0.0568 \! &\!  0.0278\!\\
   \! 0.0673 \! & \! 0.0555 \! &\!  0.0133\!   & \!0.0352 \!& \!  0.0085\!  & \! 0.1113\!\\
    \!0.0446\!  &\!  0.0117\! & \!  0.0421 \! & \! 0.0085\!  &\!  0.0260 \! &  \!0.0307\!\\
             \end{array}
           \right],\nonumber\\
&&\hspace{-8mm}\mathbf{C}=\left[
\begin{array}{cccc}
               0.0464 &   0.0335&    0.0502  &  0.0979\\
    0.0995   & 0.0962   & 0.0632 &   0.0585\\
    0.0535  &  0.0492  &  0.0922 &   0.0962\\
    0.0609   & 0.0392   & 0.0300  &  0.0335\\
    \end{array}
    \right].
\end{eqnarray}
 Following the same steps as those in Example 1, we obtain that $h^*=0$, which means $\text{Sim}_Y(Z\rightarrow X)$ holds. Furthermore, by setting ${\mathbf e}={\mathbf 1}_{24\times 1}$ in~\eqref{eq:xy16}, we obtain one possible $P_{\bar X|Z}$, denoted by matrix $\mathbf Q^*$:
\begin{eqnarray}
\mathbf Q^*=\left[\begin{array}{cccc}
    0.4979   & 0.1504 &   0.2038 &   0.1479\\
    0.0148   & 0.3751  &  0.5618  &  0.0483\\
    0.5210   & 0.4391  &  0.0254  &  0.0144\\
    0.1302   & 0.0917  &  0.0301  &  0.7481\\
    0.5638   & 0.2674  &  0.0161  &  0.1527\\
    0.0261   & 0.0622  &  0.4110  &  0.5006\\
\end{array}
\right].
\end{eqnarray}

One can easily check that $\mathbf A\mathbf Q^*=\mathbf C$ holds. We note that, because of the lack of special data structure and the high dimensions, it is difficult to use the mass constellation structure of~\cite{Maurer:TIT:031} to check whether $\text{Sim}_Y(Z\rightarrow X)$ holds or not in this example.

\textbf{Example 4:} In this example, we consider the following PMF $P_{XY}$:
\begin{eqnarray}
P_{XY}(x,y)=
\begin{cases}
\frac{1-\alpha}{2}, \quad \text{ if } x=y; \\
\frac{\alpha}{2}, \quad \quad \text{ if } x\not=y,
\end{cases}\nonumber
\end{eqnarray}
and $Z$ is generated by $[X,Y]$ via an erasure channel with erasure probability $1-\gamma$, i.e., $Z=(X,Y)$ with a probability $\gamma$ and $Z=\phi$ with probability $1-\gamma$. It was shown in~\cite{Maurer:TIT:031} that $\text{sim}_Y(Z\rightarrow X)$ and $\text{sim}_X(Z\rightarrow Y)$ hold if and only if $\gamma\ge 1-2\alpha$. In the following, we use our algorithm to verify the obtained result.

As above, in step 1, we compute $P_{YZ}$ and write $P_{YZ}$ and $P_{YX}$ in matrix form $\mathbf A$ and $\mathbf C$:
\begin{eqnarray}
\textbf{A}&=&\left[
             \begin{array}{ccccc}
               \frac{(1-\alpha)\gamma}{2} & \frac{\alpha\gamma}{2} &0&0& \frac{1-\gamma}{2} \\
               0 &0&\frac{\alpha\gamma}{2} & \frac{(1-\alpha)\gamma}{2}& \frac{1-\gamma}{2}\\
             \end{array}
           \right],\nonumber
          \\\nonumber
           \textbf{C}&=&\left[
                                \begin{array}{cc}
                                  \frac{1-\alpha}{2} & \frac{\alpha}{2} \\
                                  \frac{\alpha}{2} & \frac{1-\alpha}{2} \\
                                \end{array}
                              \right].
\end{eqnarray}
In step 2, we calculate matrices $\mathbb A$ and $\mathbf c$:
\begin{eqnarray}
&&\hspace{-8mm} \mathbb A=\nonumber \\&&\hspace{-10mm}{{\left[
         \begin{array}{cccccccccc}
           \!\!\frac{(1-\alpha)\gamma}{2}\!\!&\!\!0\!\! &\!\! \frac{\alpha\gamma}{2}\!\! &\!\!0\!\!&\!\!0\!\!&\!\!0\!\!&\!\!0\!\!&\!\!0\!\!&\!\! \frac{1-\gamma}{2}\!&\!\!0\!\! \\
            \!\!0\!\!&\!\!\frac{(1-\alpha)\gamma}{2}\!\!&\!\!0\! \!&\! \!\frac{\alpha\gamma}{2}\!\!&\!\!0\!\! &\!\!0\!\!&\!\!0\!\!&\!\!0\!\!&\!\!0\!\!&\! \!\frac{1-\gamma}{2}\!\! \\
            \!\!0\!\!&\!\!0\!\! &\!\!0&\!\!0\!\!&\!\!\frac{\alpha\gamma}{2}\!\!&\!\!0\! \!&\!\! \frac{(1-\alpha)\gamma}{2}\!\!&\!\!0\!\!&\!\! \frac{1-\gamma}{2}\!\!&\!\!0\!\!\\
            \!\!0\!\!&\!\!0 \!\!&\!\!0\!\!&\!\!0\!\!\!&\!\!0\!\!&\!\!\frac{\alpha\gamma}{2}\!\!&\!\!0 \!\!&\!\! \frac{(1-\alpha)\gamma}{2}\!\!&\!\!0\!\!&\!\! \frac{1-\gamma}{2}\!\!\\
           \!\!1\!\! &\!\! 1\!\! &\!\! 0\!\! &\!\! 0\!\! &\!\! 0\!\! &\!\! 0\!\! &\!\! 0\!\! &\!\! 0\!\! &\!\! 0\!\! &\!\! 0\!\!\\
           \!\!0 \!\!& \!\!0 \!\!& \!\!1 \!\!& \!\!1\!\! & \!\!0 \!\!& \!\!0\!\! & \!\!0\!\! &\!\! 0 \!\!& \!\!0 \!\!& \!\!0\!\!\\
           \!\!0\!\! &\!\! 0\!\! &\!\! 0\!\! &\!\! 0\!\! &\!\! 1\!\! &\!\! 1\!\!& \!\!0\!\! &\!\! 0\!\! &\!\! 0\!\! &\!\! 0\!\! \\
            \!\!0\!\! &\!\! 0\!\! &\!\! 0\!\! &\!\! 0\!\!& \!\!0\!\! &\!\! 0\!\! &\!\! 1\!\! &\!\! 1\!\!& \!\!0 \!\!& \!\!0\!\!\\
            \!\!0\!\!&\!\!0\!\!&\!\!0\!\! &\!\! 0\!\! &\!\! 0\!\! &\!\! 0\!\!&\!\! 0\!\! &\!\! 0\!\! &\!\! 1\!\! &\!\! 1\!\!\\
         \end{array}
       \right]}},\nonumber
\end{eqnarray}
$$\mathbf c=[1-\alpha, \alpha, \alpha, 1-\alpha, 1, 1, 1,1,1]^T.$$

The following steps are similar to those in Examples 1 and 2. Using our algorithm, we can find that, for any given values $\alpha$ and $\gamma$, as long as $\gamma\ge 1-2\alpha$,  $h^*=0$, and the simulatability condition holds. We can also obtain a possible simulatability  channel $P_{\bar X|Z}$ that Eve may use, following the same steps as in Example 1. On the other side, if $\gamma< 1-2\alpha$, we obtained $h^*<0$, and hence the simulatability condition does not hold.

\section{Complexity Reduction} \label{sec:recom}

In Proposition~\ref{prop:1}, we show that different choices of $\mathbb A^g$ will not affect the result on whether $h^*$ equals zero or not. However, different choices of $\mathbb A^g$ may affect the amount of computation needed. Primal-dual path-following method is one of the best methods for solving LP of the following form \cite{Bazaraa:JWS:2011}:
\begin{eqnarray}
&&\min\limits_{\mathbf t}\textbf{t}^T\mathbf b\nonumber\\
\text{s.t.}&&\mathbf{t}\succeq \mathbf 0,\nonumber\\
&&\mathbf B\mathbf t=\mathbf d,\nonumber
\end{eqnarray}
in which $\mathbf B$ is a matrix of size $m\times n$. The complexity is related to the size of $\mathbf B$. In particular, in terms of $m$ and $n$, the complexity is $O((nm^2+n^{1.5}m)L)$~\cite{Monteiro:MP:1989a,Monteiro:MP:1989}. In LP~\eqref{eq:xy15} constructed in the proof of Theorem~\ref{thm:2}, $\mathbf B=(\mathbf I-\mathbb A^g\mathbb A)^T$, which is an $n\times n$ matrix, and hence the complexity is $O(n^3L)$ as mentioned in Section~\ref{sec:MR}.

%Let's look back to the linear programming~\eqref{eq:xy15}, the complexity is $O((n^3+n^{2.5})L)$. We find that we do not need to input a $n\times n$ matrix for the linear programming~\eqref{eq:xy15} if we select an appropriate $\mathbb A^g$ in~\eqref{eq:xy13}.

In the following, we show that if we choose the g-inverse of $\mathbb A$ to be $\mathbb A^+$, the Moore-Penrose inverse, the problem size can be reduced by some further transformations.
Let the SVD of $\mathbb A$ be $\mathbf U\mathbf \Sigma \mathbf V^T$. Then $\mathbb A^+=\mathbf V\mathbf \Sigma^+ \mathbf U^T$. Suppose $\text{rank}(\mathbf\Sigma_{m\times n})=r$ and set $s=n-r$.%, then we have
%we construct a new $n\times n$ matrix as follows
%\begin{eqnarray}
%\mathbf \Lambda=\left[
%  \begin{array}{ccccc}
%    1/\sigma_{1} & 0 &\cdots& 0 & \mathbf 0^T \\
%    0 & \ddots & \ddots&\vdots & \vdots \\
%    \vdots&\ddots&\ddots&0&\vdots\\
%    0&\cdots&0&1/\sigma_{r}&\mathbf 0^T\\
%    \mathbf 0 & \cdots& \cdots&\mathbf 0 & \mathbf I_{s} \\
%  \end{array}
%\right],
%\end{eqnarray}
%where $\sigma_i,i=1,\cdots, r$ are the singular values of $\mathbb A$.
%Then we have
%\begin{eqnarray}
%\mathbb A=\mathbf U \mathbf\Sigma\mathbf V^T,\\
%\mathbb A^g=\mathbf V\mathbf \Sigma^+\mathbf U^{T}.\label{eq:newag}
%\end{eqnarray}
We have
\begin{eqnarray}
\mathbb A^+\mathbb A &=&\mathbf V\mathbf \Sigma^+\mathbf U^{T}\mathbf U \mathbf\Sigma\mathbf V^T\nonumber\\
&=&\mathbf V\left[
                          \begin{array}{cc}
                            \mathbf I_{r} & \mathbf 0_{r\times s} \\
                            \mathbf 0_{s\times r} & \mathbf 0_{s\times s} \\
                          \end{array}
                        \right]\mathbf V^T.
\end{eqnarray}

As discussed in the proof of Theorem~\ref{thm:2}, checking $\text{Sim}_{Y}(Z\rightarrow X)$ holds or not is equivalent to checking whether
\begin{eqnarray}
(\mathbf I-\mathbb A^+\mathbb A)\mathbf p\preceq \mathbb A^+ \mathbf c\label{eq:cond}
\end{eqnarray}
has a solution or not. We now perform some transformations on~\eqref{eq:cond}. First we have
\begin{eqnarray}
\mathbf I-\mathbb A^+\mathbb A&&\hspace{-6mm}=\mathbf V\left[
                          \begin{array}{cc}
                            \mathbf  I_{r} \!\!\!&  \!\!\!\mathbf 0_{r\times s} \!\!\! \\
                            \!\!\! \mathbf 0_{s\times r} \!\!\! & \!\!\! \mathbf I_s \!\!\! \\
                          \end{array}
                        \right]\mathbf V^T-\mathbf V\left[
                          \begin{array}{cc}
                            \!\!\! \mathbf I_{r} \!\!\! &  \!\!\!\mathbf 0_{r\times s}  \!\!\!\\
                           \!\!\!  \mathbf 0_{s\times r}  \!\!\!&  \!\!\!\mathbf 0_{s\times s} \!\!\! \\
                          \end{array}
                        \right]\mathbf V^T \nonumber\\
                        &&\hspace{-6mm}=\mathbf V\left[
  \begin{array}{cc}
    \textbf 0_{r\times r} & \textbf 0_{r\times s} \\
    \textbf 0_{s\times r} & \textbf I_{s} \\
  \end{array}
\right]\mathbf V^T.
\end{eqnarray}

Hence, \eqref{eq:cond} is equivalent to
\begin{eqnarray}
\mathbf V\left[
  \begin{array}{cc}
    \textbf 0_{r\times r} & \textbf 0_{r\times s} \\
    \textbf 0_{s\times r} & \textbf I_{s} \\
  \end{array}
\right]\mathbf V^T\mathbf p\preceq \mathbb A^+ \mathbf c.\label{eq:36}
\end{eqnarray}
 $\mathbf V$ can be split into four blocks as
\begin{eqnarray}
\mathbf V=\left[
                             \begin{array}{cc}
                               \mathbf V_{r\times r} & \mathbf V_{r\times s} \\
                               \mathbf V_{s\times r} & \mathbf V_{s\times s} \\
                             \end{array}
                           \right].
\end{eqnarray}
We use $\mathbf w$ to denote the $n\times 1$ column vector $\mathbf V^T\mathbf p$, i.e.,
\begin{eqnarray}
\mathbf w=\mathbf V^T\mathbf p.
\end{eqnarray}
Note that $\mathbf p \leftrightarrow\mathbf w$ is a reversible bijection, since $\mathbf V^T$ is a full rank matrix.\\
Then \eqref{eq:36} is equivalent to
\begin{eqnarray}
\left[
  \begin{array}{cc}
    \mathbf 0_{r\times r} &  \mathbf V_{r\times s} \\
     \mathbf 0_{s\times r} & \mathbf V_{s\times s} \\
  \end{array}
\right]\left[
         \begin{array}{c}
            \mathbf w_{r\times 1} \\
            \mathbf w_{s\times 1} \\
         \end{array}
       \right]\preceq \mathbb A^+ \mathbf c,
\end{eqnarray}
which is equivalent to
\begin{eqnarray}
       \left[
         \begin{array}{c}
            \mathbf V_{r\times s} \\
          \mathbf V_{s\times s} \\
         \end{array}
       \right]\left[
                \begin{array}{c}
                  \mathbf w_{s\times 1} \\
                \end{array}
              \right]\preceq \mathbb A^+ \mathbf c.   \label{eq:41}
\end{eqnarray}
Hence, checking whether~\eqref{eq:cond} has a solution or not is equivalent to checking whether~\eqref{eq:41} has a solution or not.
To check whether~\eqref{eq:41} has a solution or not, we can construct a new LP for~\eqref{eq:41} in the same way as in the proof in Theorem~\ref{thm:2}. However, the size of the newly constructed LP will be smaller than that of~\eqref{eq:xy15} constructed in the proof of Theorem~\ref{thm:2}. The complexity for the newly constructed LP will be $O((ns^2+n^{1.5}s)L)$. Since $s$ is always less than or equal to $n$ (sometimes, $s$ can be much less than $n$) and that $L$ doesn't change, compared with the LP~\eqref{eq:xy15}, the computational complexity for this new LP will be reduced. %Take the first two examples in section~\ref{sec:eple} for instance, in example 1, $n=6,\text{ }s=4$, and $n=20,\text{ }s=9$ in example 2.

\section{Conclusion}\label{sec:conc}

In this paper, we have proposed an efficient algorithm to check the simulatability condition, an important condition in the problems of secret key generation using a non-authenticated public channel. We have also proposed a simple and flexible method to calculate a possible simulatability channel if the simulatability condition holds. The proposed algorithms have polynomial complexities. We have presented numerical examples to show the efficiency of the protocol. Finally, we have proposed an approach to further reduce the computational complexity.

%In fact, this protocol can also be used for multiple terminals problem. For example, we need to check $P_{X_1X_2X_3X_4}=\sum_{\mathcal Z}P_{X_1X_2Z}\cdot P_{\bar X_3\bar X_4|Z}$, if we take $X_1\times X_2$ as a new variable $X$ and take $X_3\times X_4$ as variable $Y$, then the problem becomes the same as that of two terminals. And we can follow the same next steps as the protocol to get a reliable result.

\appendices
\section{Farkas' Lemma}\label{app:farkas}
There are several equivalent forms of the Farkas' lemma\cite{Alex:san:1998}. Here, we state a form that will be used in our proof.
\begin{lem}{(Farkas' Lemma \cite{Alex:san:1998})} \label{thm:3}
Let $\mathbf B$ be a matrix, and $\mathbf b$ be a vector, then the system specified by $\mathbf B\mathbf{p}\preceq \mathbf b$, has a solution $\mathbf{p}$, if and only if $\mathbf{t}^T\mathbf b\ge 0$ for each column vector $\mathbf{t}\succeq \mathbf 0$ with $\mathbf B^T\mathbf{t}=\mathbf0$.
\end{lem}
\bibliographystyle{ieeetr}
%\bibliography{macros,sensornetwork,secrecy}

\end{document}